\documentclass[11pt]{article}
\usepackage{amsmath}
\usepackage{amssymb}
\usepackage{slashed}
\usepackage{microtype}
\usepackage{amsthm}
\usepackage[margin=2.5cm]{geometry}
\usepackage[numbers, sort&compress]{natbib}

\newtheorem{proposition}{Proposition}[section]
\newtheorem{corollary}[proposition]{Corollary}
\newtheorem{hypothesis}{Hypothesis}
\theoremstyle{definition}
\newtheorem{definition}[proposition]{Definition}
\theoremstyle{remark}
\newtheorem{remark}[proposition]{Remark}

\title{A structural criterion for asymptotic states in Supersymmetry}

\author{
	Stefano Bellucci$^{1,2,*}$ and Stefania De~Matteo$^{3}$
}

\date{}
\begin{document}
	\maketitle

	\begin{center}
		{\small
			$^{1}$INFN--Laboratori Nazionali di Frascati, Frascati, Italy\\
			$^{2}$Ecotec University, Samborond\'on, Ecuador\\
			$^{3}$Department of Mathematics and Physics, Roma Tre University, Rome, Italy\\[0.5em]
			$^{*}$Corresponding author: stefano.bellucci@lnf.infn.it
		}
	\end{center}

\begin{abstract}
\noindent
The occurrence of a field in a quantum field theory does not by itself establish the existence of a corresponding asymptotic particle. We formulate a structural criterion that separates operator admissibility, dynamical spectral realization, and particlehood, using the exterior-polynomial classification of gravitino mass bilinears as the operator-level prototype. Particlehood is organized as the endpoint of an ordered filtration of geometry and state, effective algebra, dynamics, and spectral analysis, and the localization criterion $\Lambda[\Psi]$ is identified with the presence of a physical spectral atom of non-zero weight in the relevant spectral construction. The framework yields four concrete results. First, conventional linearly realized global supercharges are obstructed on a positive-definite de Sitter Hilbert space whenever the full diagonal odd--odd bracket satisfies the stated positivity-reversal hypothesis; superconformal realizations fall outside that algebraic domain. Second, invariant scalar de Sitter product decompositions remove the Minkowski additive two-particle threshold in principal-series channels, so threshold protection cannot be imported unchanged from flat space. Third, a spectral-transfer decomposition shows explicitly how continuum weight in a mediator propagator can open absorptive support for an external field below the atomic mediator threshold. Fourth, in a reduced effective-operator basis, the fermionic vector derivative coupling is redundant under current-conservation assumptions, while the axial dimension-five coupling is physical unless excluded by parity or $CP$; under that symmetry hypothesis the displayed residual fermionic derivative potential starts one operator dimension above the scalar one. Exact unbroken supersymmetry enforces multiplet rigidity of the atomic criterion. These results define a testable structural sieve for asymptotic states without identifying geometry alone with particle loss.
\end{abstract}

\noindent\textbf{Keywords:} quantum field theory; supersymmetry; asymptotic states; localization criterion; spectral measure; K\"all\'en--Lehmann representation; gravitino; exterior-polynomial classification; infraparticles

\section{Introduction}
\label{sec:intro}

In relativistic quantum field theory, a field and a particle are not interchangeable notions. Fields enter the theory as operator-valued distributions, whereas the particle interpretation requires additional spectral and asymptotic conditions. In the LSZ framework, a stable one-particle state is associated with an isolated mass-shell contribution of non-zero residue. Haag's theorem, algebraic quantum field theory, and the infraparticle problem all show that formal field content need not map one-to-one onto asymptotic particle content~\cite{1,2,3,4,5,6,7,8}.

Supersymmetry makes this distinction especially sharp. Its algebra relates bosonic and fermionic fields and organizes them into multiplets~\cite{9,10,11,12,13}. Yet the algebraic presence of a component does not by itself establish that the corresponding excitation possesses an isolated pole, survives infrared dressing, or can be prepared and detected as an asymptotic particle. In supergravity the issue is more delicate because local supersymmetry is a gauge symmetry and the physical state space depends on constraints, gauge invariance, and the chosen background~\cite{14,15,16}.

We formulate a binary localization criterion $\Lambda[\Psi]$ that distinguishes algebraically admissible fields from configurations capable of forming stable, finite-energy, phase-coherent asymptotic states, and we examine the asymmetry between fermionic and scalar responses to slowly varying structural fluctuations. The scope is explicit: first-order fermionic structure yields a channel-dependent operator hierarchy rather than universal protection; ensemble dephasing is distinguished from an intrinsic decay width; and any Standard-Model--superpartner split must arise from the realized symmetry, masses, couplings, and channel kinematics rather than from geometry alone.

The construction is modelled on the exterior-polynomial classification of the gravitino mass sector~\cite{17}, which distinguishes: (i) classification of admissible structures in a specified algebra; (ii) physical selection by the kinetic and constraint structure together with local supersymmetry; and (iii) determination of the coefficient by the background, supersymmetry breaking, and matter couplings. The present paper transposes this same layering from operators to states.

The central claim is therefore limited and structural. Particlehood is the endpoint of an ordered sequence:
\begin{equation*}
\text{geometry and state} \longrightarrow \text{effective algebra} \longrightarrow \text{dynamics} \longrightarrow \text{spectral realization} \longrightarrow \text{particlehood}.
\end{equation*}
The first stages determine what can be consistently posed; the dynamical stage determines the dressed propagator; the spectral stage determines whether a physical atom survives in the applicable construction. No stage alone is sufficient.

The paper proceeds as follows. Section~\ref{sec:operator-state} explains the correspondence between operator admissibility in the gravitino classification and state admissibility. Section~\ref{sec:filtration} introduces the reduced filtration. Section~\ref{sec:order} clarifies its order dependence on de Sitter backgrounds and states the positivity obstruction for the specified conventional global superalgebra. Section~\ref{sec:spectral} gives the localization criterion a measure-theoretic formulation, in Minkowski space and on de Sitter space, where the absence of a Minkowski additive threshold in specified scalar channels changes the spectral analysis. Section~\ref{sec:channels} distinguishes direct from mediated spectral susceptibility: it gives the geometric origin of the leading derivative-channel comparison, the discrete-symmetry condition under which it holds, and a worked example of spectral transfer. Section~\ref{sec:split} establishes the conditions under which an asymmetric realized spectrum is structurally possible. The discussion states the limits of the framework and its relation to the companion dynamical study. Flat-space formulas use the metric signature $(+,-,-,-)$, so that $p^2=m^2$ on the mass shell.

\section{From operator admissibility to state admissibility}
\label{sec:operator-state}

The gravitino analysis of Ref.~\cite{17} separates two classification problems. In the unrestricted derivative-free vector--spinor algebra, inverse coframes, contractions of the vector indices, interior products, and Hodge duals are available. The parity-even sector therefore contains independent structures such as $\bar\psi_a\psi^a$ and $(\bar\psi_a\gamma^a)(\gamma^b\psi_b)$, with
\begin{equation*}
\bar\psi_a\gamma^{ab}\psi_b = (\bar\psi_a\gamma^a)(\gamma^b\psi_b) - \bar\psi_a\psi^a .
\end{equation*}
Consequently, Lorentz covariance, algebraicity, and parity do not uniquely determine the Rarita--Schwinger contraction in the full algebra.

A different statement holds in the restricted first-order Cartan sector generated polynomially by the coframe $e^a$, the Majorana vector--spinor one-form $\psi$, its conjugate, gamma matrices, and wedge multiplication, while excluding inverse coframes, interior products, and Hodge duals acting directly on $\psi$. After scalar coefficient functions are factored out, the real tensor-structure space is two-dimensional, with a manifestly real action-form basis that may be written as $B_-=\bar\psi\wedge\gamma^{(2)}\wedge\psi$ and $R_+=i\bar\psi\wedge\gamma_5\gamma^{(2)}\wedge\psi$ in the conventions of Ref.~\cite{17}. Restoring arbitrary allowed Lorentz-scalar coefficient functions gives a rank-two module, and parity selects the conventional Rarita--Schwinger mass line. This is a classification theorem for the stated exterior-polynomial algebra; it is not a uniqueness theorem for the unrestricted vector--spinor problem.

The classification does not by itself select the complete physical theory. The Rarita--Schwinger combination realized in supergravity is tied to the kinetic operator, constraint propagation, and local supersymmetry. Its coefficient is fixed separately by the background and supersymmetry-breaking data, for example by the AdS scale or the K\"ahler-covariant mass function in matter-coupled $N=1$ supergravity. Thus the operator-level result has three layers:
\begin{equation}
\text{structural admissibility} \longrightarrow \text{dynamical and symmetry selection} \longrightarrow \text{coefficient determination}.
\label{eq:layers}
\end{equation}

The superspace statement is equally important. A picture-changing operator or pullback to the bosonic body extracts the relevant component of a completed super-four-form. It does not manufacture the one-form index of the gravitino, create the Clifford structure, or introduce a mass scale. Equivalence between different projection representatives requires the usual hypotheses of closedness, appropriate cohomology class, global definition or patching, and controlled boundary terms.

The present paper uses only the logical architecture of this result. At the state level, one must again separate what is admissible, what is dynamically selected, and what is physically realized. The operator classification does not imply a pole theorem and applies to a Majorana vector--spinor independently of whether it belongs to an observable supermultiplet. Its role here is to provide a precise prototype of a broader rule: classification at one structural level cannot be promoted automatically to existence at a later physical level.

\begin{proposition}[Operator--state correspondence]
The operator classification of Ref.~\cite{17} and the present state criterion instantiate the same logical pattern but are not derivations of one another. At the operator level, the chosen geometric algebra restricts the admissible bilinears before symmetry and dynamics select the physical term. At the state level, geometry and the quantum state restrict the spectral environment before dynamics determines whether an isolated particle contribution survives.
\end{proposition}

\section{A reduced particlehood filtration}
\label{sec:filtration}

Let $X_0$ denote the formal field content and writable interaction structures of a theory. We distinguish four stages:
\begin{equation}
X_0 \supseteq X_G \supseteq X_{GA} \longrightarrow X_{\mathrm{dyn}} \longrightarrow X_{\mathrm{spec}} \longrightarrow X_{\mathrm{part}} .
\label{eq:filtration}
\end{equation}
The first two restrictions specify the background-dependent domain of the physical problem. The arrows that follow are not additional abstract sieves: they denote the dynamical calculation, its spectral diagnosis, and the assignment of particlehood.

\subsection{Geometry and quantum state}

For a background metric $g_{\mu\nu}$ and a physically admissible state $\omega$, the map
\begin{equation*}
F_G[g,\omega] : X_0 \longrightarrow X_G(g,\omega)
\end{equation*}
specifies the causal and horizon structure and the state prescription, constrains the available mode and positive-frequency constructions, and sets the domain for infrared and spectral analysis. These data do not by themselves determine the interacting infrared correlator. The geometry and state therefore define the spectral arena; they do not by themselves decide the fate of every pole.

\subsection{Effective algebra}

The second restriction,
\begin{equation*}
F_A[\,\cdot\,;F_G] : X_G \longrightarrow X_{GA} ,
\end{equation*}
contains the symmetries and constraints physically realized on the state space selected by $F_G$: Lorentz and gauge representations, BRST cohomology where applicable, physical polarizations, conserved charges, and the admissible interaction vertices. This is an effective algebraic structure, not merely the abstract symmetry algebra of the classical Lagrangian.

\subsection{Dynamics, spectral decision, and particlehood}

Only after these restrictions are specified does one compute the dressed correlation functions. Kinematic thresholds enter as part of this dynamical calculation, not as an independent ontological level. Spectral analysis then asks whether the physical dressed two-point function contains an atomic contribution of non-zero weight, with an isolated massive pole as the regular case. In the present paper this is a criterion for a stable asymptotic particle state, not for an unstable resonance that may nevertheless be experimentally observable. In non-stationary settings the same language is used only within a regime in which a local or adiabatic spectral construction is justified.

\begin{hypothesis}[Ordered particlehood filtration]
For a fixed theory, background, and state prescription, the standard particle interpretation is the endpoint of the ordered sequence
\begin{equation*}
\text{geometry/state} \to \text{effective algebra} \to \text{dynamics} \to \text{spectrum} \to \text{particlehood}.
\end{equation*}
The ordering expresses dependence of the downstream problem on the upstream data; it is not a claim that all stages are linear operators acting on a common space.
\end{hypothesis}

The binary criterion is then a composite diagnostic on the domain where the physical spectral construction is defined: $\Lambda[\Psi]=1$ if the measure contains a physical atom of positive weight, and $\Lambda[\Psi]=0$ if it does not. Failure to establish the required spectral representation does not by itself establish the absence of such an atom. The geometry and algebra specify the problem and the applicability of its propositions; they do not replace the calculation of the dressed correlator or assign a value of $\Lambda$ from a background label alone.

\section{Order dependence of the two primary sieves}
\label{sec:order}

The formal field algebra can be written before a background is chosen, but the effective physical algebra can only be evaluated afterwards. The relevant claim is therefore hierarchical order dependence, not a commutator statement about two endomorphisms acting on a common space: the reverse composition need not be defined on the same domain.

\begin{proposition}[Order dependence of the primary sieves]
\label{prop:order}
Let $F_G^{\mathrm{dS}}$ denote the restriction associated with a four-dimensional de Sitter spacetime and a physically admissible state, and let $F_G^{\mathrm{M}}$ denote the restriction associated with Minkowski spacetime and its vacuum. Define the effective symmetry content of $F_A[\,\cdot\,;F_G]$ as the set of symmetries realized linearly and unitarily, with the appropriate positivity properties, on the state space produced by $F_G$. Within the class of global superalgebras whose even part consists of spacetime isometries and the specified internal generators, and for a theory whose Minkowski limit is $N=1$ supersymmetric, $F_A[\,\cdot\,;F_G^{\mathrm{M}}]$ may contain linearly realized supersymmetry, whereas $F_A[\,\cdot\,;F_G^{\mathrm{dS}}]$ does not under the same unitarity and positivity requirements. Therefore the effective algebraic sieve is not definable independently of the geometric sieve, and the construction is order-dependent.
\end{proposition}

\paragraph{Argument.}
Under the standard assumptions relevant here---a physical Hilbert-space realization with positive norm and the usual Hermiticity/positivity properties of the supercharges---four-dimensional de Sitter symmetry does not admit the conventional linearly realized $N=1$ supersymmetry familiar from Minkowski space. The classification of de Sitter superalgebras~\cite{18} exhibits the associated non-compact internal-symmetry obstruction, while explicit de Sitter supergravity constructions can instead realize supersymmetry non-linearly through constrained Volkov--Akulov multiplets~\cite{19,20}; see also Ref.~\cite{21}. Thus the symmetry content entering the effective algebra depends on the background and on the class of physical realizations admitted. This is the only conclusion required for the order-dependence statement.

A sufficient positivity obstruction can be stated as the following operator lemma. Its application to a particular superalgebra requires checking the full even element in each diagonal bracket.

\begin{proposition}[Positivity obstruction to linearly realized supercharges on de Sitter space]
\label{prop:positivity}
Let $\mathcal{H}$ be a Hilbert space with positive-definite inner product carrying a unitary representation of the connected de Sitter group, or its spin cover, with generators $M_{AB}$, and let $Q_\alpha$ be operators on $\mathcal{H}$ whose adjoints are among the odd generators of a superalgebra extending $\mathfrak{so}(4,1)$, so that the odd--odd bracket closes on the even part,
\begin{equation}
\{Q_\alpha, Q_\beta^\dagger\} = c_{\alpha\beta}^{\;AB}\,M_{AB} + r_{\alpha\beta}^{\;i}\,T_i ,
\label{eq:oddbracket}
\end{equation}
with $T_i$ the generators of a compact internal group represented unitarily on $\mathcal{H}$. The identities are understood on a common invariant dense domain, with compatible adjoints and self-adjoint realizations of the even elements. If, for every $\alpha$, the even element on the right-hand side of $\{Q_\alpha,Q_\alpha^\dagger\}$ is conjugate to its negative under the group, then $Q_\alpha=0$ on $\mathcal{H}$ for every $\alpha$.
\end{proposition}

\paragraph{Argument.}

On the common domain, the quadratic form of $X_\alpha=\{Q_\alpha,Q_\alpha^\dagger\}$ is
\begin{equation*}
\langle v,X_\alpha v\rangle
=\|Q_\alpha v\|^2+\|Q_\alpha^\dagger v\|^2\geq0.
\end{equation*}
If $U(g)X_\alpha U(g)^{-1}=-X_\alpha$, positivity is preserved by the unitary conjugation, so both $X_\alpha$ and $-X_\alpha$ are positive. Hence $X_\alpha=0$ and both supercharges vanish on the domain, with the corresponding conclusion for their closures. The hypothesis concerns the entire even element, not just an individual isometry generator. When the even group is a direct product it is sufficient that its spacetime and internal components can each be reversed by conjugation. The spacetime part of this condition can be checked completely for four-dimensional de Sitter space, as follows. The internal part remains non-automatic: conjugation leaves an abelian $U(1)$ generator unchanged, and a generic element of $SU(n)$ for $n\geq3$ need not be conjugate to its negative. The lemma gives no conclusion unless the full even element satisfies the stated hypothesis.

\paragraph{Spacetime reversal check.}
Identify $\mathfrak{so}(4,1)$ with the skew-symmetric endomorphisms of the Lorentzian vector space $\mathbb{M}^{1,4}$. The canonical-form classification of Lorentzian skew endomorphisms~\cite{42} gives, for every $X\in\mathfrak{so}(4,1)$, an orthogonal $X$-invariant decomposition
\begin{equation*}
\mathbb{M}^{1,4}=\mathbb{M}^{1,2}\oplus\Pi,
\qquad
X=X_{(1,2)}\oplus \mu J_{\Pi},
\end{equation*}
where $\Pi$ is a spacelike two-plane and the Lorentzian three-dimensional block is, including degenerate limits, of one of three types: a boost block plus a zero spacelike direction, a spatial rotation block plus a zero timelike direction, or a nilpotent null-rotation block. Choosing the corresponding canonical orthonormal basis $(e_0,e_1,e_2)$, these blocks can be represented, up to their real coefficients and harmless sign conventions, by $K_{01}$, $J_{12}$, or $K_{01}+J_{12}$, respectively. The orthochronous spatial reflection $r_3=\mathrm{diag}(1,-1,1)$ reverses each of these canonical blocks. Thus in every case there exists an orthochronous $r_3\in O^+(1,2)$ with $\det r_3=-1$ and
\begin{equation*}
r_3 X_{(1,2)} r_3^{-1}=-X_{(1,2)}.
\end{equation*}
On $\Pi$, a reflection $r_2\in O(2)$ has $\det r_2=-1$ and satisfies $r_2J_{\Pi}r_2^{-1}=-J_{\Pi}$; if $\mu=0$ any reflection may be used. Hence $g=r_3\oplus r_2$ preserves time orientation, has determinant $+1$, belongs to $SO_0(4,1)$, and obeys
\begin{equation}
\operatorname{Ad}(g)X=-X
\qquad\text{for every }X\in\mathfrak{so}(4,1).
\label{eq:spacetime-reversal}
\end{equation}
The same adjoint action is obtained after lifting $g$ to the spin cover. This argument covers the semisimple, nilpotent, and mixed semisimple--nilpotent adjoint classes at once; it does not remove the independent internal-generator condition in Proposition~\ref{prop:positivity}.

\begin{remark}[Scope of the obstruction]
The lemma isolates a sufficient obstruction for the specified positive-norm global operator realization. It is not a prohibition of every supersymmetric construction on de Sitter space. The classifications in Refs.~\cite{18,22} concern particular superalgebras and reality conditions. Superconformal constructions instead include conformal generators in the even algebra and can contain a positive conformal charge~\cite{23}; they lie outside the assumed closure. The lack of a globally timelike de Sitter Killing vector motivates the contrast with the Minkowski algebra, where the sum of the diagonal supercharge anticommutators is proportional to $P^0\geq0$, but is not a substitute for checking the operator hypotheses. Similarly, a nonlinear realization or $Q_\alpha|0\rangle\neq0$ does not by itself evade this lemma, which does not assume vacuum invariance: the change in its global operator or algebraic hypotheses must be identified. Nonlinear supergravity constructions~\cite{19,20} and the physical algebra on a chosen background must therefore be distinguished from the global realization specified above.
\end{remark}

\begin{remark}[Geometric counterpart: Killing spinors]
The obstruction has a classical counterpart in the Killing spinor equation. For the ordinary Killing-spinor realization considered in this remark, a preserved supercharge requires a spinor $\epsilon$ with $\nabla_\mu\epsilon = \tfrac{\lambda}{2}\gamma_\mu\epsilon$. In this remark we use the signature $(-,+,+,+)$ of Ref.~\cite{17}, in which the gamma matrices admit a real Majorana representation. Iterating the equation gives $[\nabla_\mu,\nabla_\nu]\epsilon = -\tfrac{\lambda^2}{2}\gamma_{\mu\nu}\epsilon$, while on a maximally symmetric space with $R_{\mu\nu\rho\sigma}=k\,(g_{\mu\rho}g_{\nu\sigma}-g_{\mu\sigma}g_{\nu\rho})$ one has $[\nabla_\mu,\nabla_\nu]\epsilon = \tfrac14 R_{\mu\nu\rho\sigma}\gamma^{\rho\sigma}\epsilon = \tfrac{k}{2}\gamma_{\mu\nu}\epsilon$. Integrability therefore requires $k=-\lambda^2$. Anti de Sitter space, $k<0$, admits a real $\lambda$ and real Killing spinors; de Sitter space, $k=H^2>0$, forces $\lambda=\pm iH$, and in the real representation a solution cannot be real, so no Killing spinor obeys the Majorana condition. This is the geometric obstruction for the particular ordinary Killing-spinor equation written here. It does not exclude Majorana conformal Killing spinors or other generalized equations with additional background fields; the superconformal distinction is explicit in Ref.~\cite{23}.
\end{remark}

\begin{remark}[Order dependence and non-exchangeability]
Proposition~\ref{prop:order} does not require the formal construction of $F_G\circ F_A$. It states that $F_A$ is conditioned by the output of $F_G$ and that reversing the conceptual order changes, or leaves undefined, the physical construction. The primary sieves are therefore non-exchangeable without any assertion that they form two operators on a common domain whose commutator can be evaluated.
\end{remark}

\begin{remark}[No double counting]
Proposition~\ref{prop:order} also resolves the potential overlap between the two sieves. Structures that could be attributed to either---the notion of positive frequency, the infrared sector, background-induced symmetry breaking---are assigned by convention as follows: the state and its microlocal properties belong to $F_G$; algebraic relations, constraints, cohomology, and the inventory of realized symmetries belong to $F_A$, evaluated on the output of $F_G$. Since $F_A$ is defined relative to $F_G$, no structural effect is counted twice.
\end{remark}

\section{Spectral realizability and the localization criterion}
\label{sec:spectral}

\subsection{The spectral measure of the dressed two-point function}

For a scalar physical interpolating operator in a Poincar\'e-invariant vacuum with a positive Hilbert space and the spectrum condition, the dressed two-point function admits the K\"all\'en--Lehmann representation~\cite{24,25}
\begin{equation}
\Delta_\Psi(p^2) = \int_0^\infty \frac{d\rho_\Psi(s)}{p^2 - s + i\epsilon} ,
\label{eq:KL}
\end{equation}
where $d\rho_\Psi$ is positive. Necessary subtractions or local contact terms, for example for composite operators, are understood separately and do not create physical atoms. Stationarity alone is insufficient for the Lorentz-invariant form~\eqref{eq:KL}; in a medium the spectral function can depend separately on energy and spatial momentum. A generic cosmological background need not admit this representation. A local or adiabatic construction is used only in a controlled regime, with its approximation and physical spectral interpretation stated~\cite{26}. The invariant de Sitter representation below is a different harmonic construction.

For gauge fields and the gravitino, positivity refers to correlators of operators acting in the positive physical state space. A covariant-gauge propagator is not assumed positive; invoking BRST cohomology requires identifying the physical operator or projection actually used. Spinor and tensor correlators require the corresponding spectral matrix. In the parity-preserving Dirac case,
\begin{equation*}
S(p)=\int_0^\infty
\frac{\slashed p\,d\rho_1(s)+d\rho_2(s)}{p^2-s+i\epsilon}.
\end{equation*}
Positivity constrains the matrix rather than requiring both coefficient measures to be independently positive. A conventional one-particle contribution has weights $Z\delta_{m^2}$ in $d\rho_1$ and $mZ\delta_{m^2}$ in $d\rho_2$~\cite{27}. At $m=0$ the latter vanishes, without removing the physical atom. Chiral or parity-violating theories can require additional structures.

\begin{definition}[Spectral classes]
\label{def:spectral}
On the domain where the physical spectral representation is available, decompose its measure into atomic, absolutely continuous, and singular-continuous parts,
\begin{equation}
d\rho_\Psi(s)=\sum_a Z_a\,\delta_{m_a^2}(ds)
+\tilde\rho(s)\,ds+d\rho_{\mathrm{sc}}(s),
\qquad Z_a>0.
\label{eq:decomposition}
\end{equation}
For spinning fields this notation refers to the appropriate physical spectral contribution. The localization criterion is
\begin{equation}
\Lambda[\Psi]=1
\quad\Longleftrightarrow\quad
d\rho_\Psi\text{ contains a physical spectral atom of weight }Z>0,
\label{eq:criterion}
\end{equation}
and is zero when no such atom is present. Its interpretation as an asymptotic particle criterion is restricted to the assumptions of the corresponding asymptotic construction.
\end{definition}

The companion analysis of Ref.~\cite{41} employs the narrower isolated-pole/LSZ-like subclass of this criterion in its dynamical implementation. Whenever the physical atom is isolated the two conventions coincide; Eq.~\eqref{eq:criterion} is stated in the more general measure-theoretic form in order to include, for example, massless atoms at a continuum edge and exactly protected embedded atoms.

An isolated massive atom is the regular pole case of this criterion. An atom at a continuum edge, or protected against coupling to an overlapping continuum, must be distinguished from a resonance with no exact point mass. In particular, the atom at $s=0$ of a massless particle need not be separated by a gap from massless multiparticle support; an infraparticle instead lacks the point mass~\cite{4,8}. Nonzero spectral weight in an accessible channel can generically destabilize an embedded atom under the additional on-shell conditions stated in Section~\ref{sec:channels}, but support overlap alone does not erase an atom by definition. A singularity without point mass, a purely continuous physical measure, and the absence of a physical contribution in a gauge-dependent description are distinct diagnoses. Where no physical spectral construction has been established, no value of~\eqref{eq:criterion} follows from that failure alone.

\subsection{Dephasing versus dissipation}
\label{sec:dephasing}

A structural clarification is in order concerning stochastic modelling. When correlators are averaged over an ensemble of classical background configurations, the averaged amplitude generically decays even though each realization evolves unitarily; this is inhomogeneous broadening (dephasing) of the averaged correlator, not an intrinsic imaginary self-energy of a fixed theory. The criterion~\eqref{eq:criterion} refers to the spectral measure of the theory in a fixed state. Stochastic background models, including the one employed in Section~\ref{sec:channels}, are therefore proxies: their damping rates estimate genuine self-energy effects only under an open-system, in-in treatment in which the environment is integrated out with its own dynamics~\cite{28,29,30}. This distinction does not invalidate the structural comparison below, but it fixes its logical status.

\subsection{The spectral sieve on a de Sitter background}
\label{sec:dS-spectral}

The invariant de Sitter spectral construction must be distinguished from the momentum-space formula~\eqref{eq:KL}. We use a positive physical Hilbert space, an invariant analytic state of the Euclidean (Bunch--Davies) class, and the representation-theoretic and analyticity hypotheses appropriate to the chosen correlator~\cite{31,32}. Geometry alone does not imply these state assumptions. For a scalar on four-dimensional de Sitter space, its principal-series contribution has the form
\begin{equation}
W_\Psi(x,y)=\int d\rho_{\Psi,\mathrm{pr}}(\nu)\,W_\nu(x,y)
+W_{\Psi,\mathrm{other}}(x,y),
\qquad m_\nu^2=H^2\left(\tfrac94+\nu^2\right),
\label{eq:KL-dS}
\end{equation}
where the mass is the effective scalar mass, including any curvature coupling. The last term denotes other representation contributions when required; its measure need not be universally discrete. In specific free-field product decompositions, complementary-series contributions arise as discrete residues~\cite{33,32}. Other spins require their own harmonic kernels and mass parametrizations~\cite{32,34,35}. The flat limit of the free kernels and of the spectral measure must be controlled together before identifying~\eqref{eq:KL-dS} with~\eqref{eq:KL}.

Within this physical harmonic construction the atomic diagnostic reads
\begin{equation}
\Lambda_{\mathrm{dS}}[\Psi]=1
\iff d\rho_\Psi\text{ contains a physical representation atom with }Z>0.
\label{eq:criterion-dS}
\end{equation}
This is a spectral statement; it does not by itself establish Minkowski-type in/out states on de Sitter space.

\begin{proposition}[Principal-series threshold statement on de Sitter space]
\label{prop:dS-sieve}
(i) In the invariant analytic setting considered above, the spectral measure of the product of two free scalar principal-series two-point functions has support throughout the principal series~\cite{33,32}. This channel has no Minkowski additive mass threshold.
(ii) An external scalar in that series is therefore not protected against this channel by that threshold alone. A nonzero absorptive contribution additionally requires a nonvanishing physical coupling and an appropriate self-energy construction. The product decomposition alone is not a nonperturbative theorem excluding every atom of the external interacting field.
(iii) In the scalar perturbative decay calculations for channels closed in flat space, the curvature-induced rate is exponentially suppressed as $H\to0$~\cite{33,36}. Recovery of an atom in a limiting measure additionally requires control of its mass, positive weight, and spectral convergence. No replacement of the fixed-background diagnostic by a limiting diagnostic is implied.
\end{proposition}

\paragraph{Argument and scope.}
The support statement follows from the scalar product decompositions in Refs.~\cite{33,32}. The absence of a globally timelike de Sitter Killing vector explains why the Minkowski energy-threshold argument cannot simply be imported, but does not replace that decomposition. To pass from the channel measure to a decay rate or external spectral conclusion one must retain the specified interaction, state, representation, and perturbative order. Cancellations and vanishing physical matrix elements remain relevant. Complementary-series and spinning channels require their own analysis.

In particular, Ref.~\cite{35} provides spinor product decompositions but distinguishes products of Wightman functions from time-ordered self-energies. Its general extension of the de Sitter self-energy spectral equations to spinor loops is formulated as a hypothesis, and the additional pole contributions that may be needed for spinor--scalar products with light scalars are left for separate study. Those results cannot be used as an unrestricted theorem for all fermionic channels.

\begin{remark}[A vanishing rate and recovery of an atom]
A family of narrow spectral densities can converge to an atom if its peak mass converges and its integrated peak weight tends to $Z_*>0$. If that weight tends to zero, narrowing alone does not leave a physical atom. Consequently $\lim_{H\to0}\Lambda[\rho_H]$ need not equal $\Lambda[\lim_{H\to0}\rho_H]$: the atomic diagnostic is not continuous under weak convergence. Moreover a family of invariant states at different constant values of $H$ is not the evolution of one state on a background $H(t)$. The latter is the separate persistence problem described in Section~\ref{sec:persistence}. No universal asymmetric particle spectrum follows from the scalar support result.
\end{remark}

\section{Direct and mediated spectral susceptibility}
\label{sec:channels}

We now distinguish two ways in which the spectral measure of a field can be degraded: direct susceptibility, through the field's own coupling to the structural environment, and mediated susceptibility, through interactions with fields whose spectral measure has already been degraded.

\subsection{Direct channel: the structural background model}
\label{sec:direct}

We employ an effective background field $\Xi(x)$ representing slowly varying structural fluctuations of the environment---a parametric proxy for infrared and environmental effects already known in quantum field theory, carrying no conserved charges and representing no new force or fundamental medium~\cite{37,28,29}. We assign $\Xi$ mass dimension one and take it to vary on scales much larger than the Compton wavelengths of the massive fields considered. The following comparison concerns a specified operator sector and basis, not a complete classification of every interaction allowed in a general effective field theory.

For a Dirac fermion the effective equation is
\begin{equation}
(i\gamma^\mu\partial_\mu - m - g_f\Xi)\,\psi = 0 ,
\label{eq:dirac}
\end{equation}
and for a scalar
\begin{equation}
\left(\Box + m^2 + g_s\Xi^2 + \kappa(\partial_\mu\Xi)^2\right)\phi = 0 .
\label{eq:scalar}
\end{equation}
A linear coupling $\mu\,\Xi\,\phi^\dagger\phi$, of dimension three, is also admissible; it adds to the non-derivative mass noise, in the same way as $g_f\Xi$ for the fermion, and does not enter the comparison of derivative channels below. From the effective-field-theory viewpoint, the operator $(\partial_\mu\Xi)^2\phi^\dagger\phi$ has mass dimension six, with $\kappa=c_\kappa/M_*^2$ for a structural suppression scale $M_*$. The corresponding fermionic operator $(\partial_\mu\Xi)^2\bar\psi\psi$ has dimension seven, suppressed by an additional power of $M_*$ and tied to the chirality-breaking structure $\bar\psi\psi$. Within this displayed operator sector, and under the symmetry and basis assumptions stated below, a residual scalar derivative potential can occur at order $M_*^{-2}$ while the displayed fermionic potential starts at order $M_*^{-3}$. Writing $\Xi(t)=\Xi_0+\delta\Xi(t)$ with a stationary random process of variance $\sigma_\Xi^2$ and correlation time $\tau_c$, a short-correlated Gaussian-noise treatment gives a model-dependent dephasing law of the schematic form
\begin{equation}
\langle\phi(t)\rangle \sim e^{-imt}\,e^{-\Gamma_s t},
\qquad
\Gamma_s \propto \frac{\kappa^2\,\sigma_{(\partial\Xi)^2}^2\,\tau_c}{m^2} ,
\label{eq:gamma-s}
\end{equation}
where $\sigma_{(\partial\Xi)^2}^2$ is the variance of the derivative noise $(\partial_\mu\Xi)^2$ to which the operator couples, and the numerical coefficient depends on the normalization and correlator assigned to $\Xi$. Equation~\eqref{eq:gamma-s} is used only as an illustrative scaling relation, not as an intrinsic decay width of the quantum field.

\paragraph{Geometric origin of the channel suppression.}
The distinction between a removable connection and a residual potential provides a geometric interpretation of this channel comparison, subject to the stated symmetry and effective-field-theory assumptions. The derivative sector is understood throughout as the sector generated by the one-form $d\Xi$: operators linear in $\partial_\mu\Xi$ and operators built from its quadratic invariant $(\partial_\mu\Xi)^2$. Operators containing second derivatives of the background, $(\Box\Xi)\,\phi^\dagger\phi$ of dimension five and $(\Box\Xi)\,\bar\psi\psi$ of dimension six, are not of this form: they multiply the mass operator of the respective field and are potentials for both, so they belong to the mass-noise class alongside $g_f\Xi$, $g_s\Xi^2$ and $\mu\Xi$, with $\Box\Xi$ in place of $\Xi$ as noise variable, suppressed by $1/M_*$ for the scalar and by $1/M_*^2$ for the fermion. They can contribute to the ensemble dephasing of both sectors. This separation is a convention for the displayed basis: integrations by parts and field redefinitions can move terms between these descriptions, and all induced contributions must be retained at the order being compared~\cite{38}. The fermionic operators of dimension five that are linear in $\partial_\mu\Xi$,
\begin{equation}
\mathcal{O}_5 = \frac{c_S}{M_*}\,(\partial_\mu\Xi)\,\bar\psi\gamma^\mu\psi ,
\qquad
\mathcal{O}_5' = \frac{c_A}{M_*}\,(\partial_\mu\Xi)\,\bar\psi\gamma^\mu\gamma_5\psi ,
\label{eq:O5}
\end{equation}
lie one order below the dimension-six scalar operator and would, if effective, provide a fermionic derivative channel at lower order than the scalar one. Both are connections, for the vector and for the chiral phase rotations of the spinor field, but they act differently. Adding $\mathcal{O}_5$ to the Dirac Lagrangian amounts to the covariant derivative
\begin{equation}
D_\mu = \partial_\mu - i\lambda_S\,S_\mu ,
\qquad
S_\mu = \partial_\mu\Xi ,
\qquad
\lambda_S = \frac{c_S}{M_*} ,
\label{eq:connection}
\end{equation}
so that at this order the background couples to the fermion through the one-form $S=d\Xi$ acting as an abelian connection on the spinor bundle. The connection is exact and therefore flat,
\begin{equation}
F_{\mu\nu} = \partial_\mu S_\nu - \partial_\nu S_\mu = 0 ,
\label{eq:flat}
\end{equation}
its parallel transport from $x_0$ to $x$ is the path-independent phase $\exp\{i\lambda_S[\Xi(x)-\Xi(x_0)]\}$, and its holonomy vanishes on every closed loop in a simply connected region. Equivalently, $\mathcal{O}_5$ is removed by the local phase rotation $\psi\to e^{i\lambda_S\Xi}\psi$; on shell it is a total derivative, $\mathcal{O}_5 = \lambda_S\,\partial_\mu(\Xi J^\mu) - \lambda_S\,\Xi\,\partial_\mu J^\mu$ with $J^\mu=\bar\psi\gamma^\mu\psi$ conserved. This removal assumes that the current is conserved and non-anomalous, that the transformation respects the boundary and global conditions, and that the remaining interactions are invariant or transformed consistently. Under these conditions the term adds no independent physical decay mechanism. Correlators of the original interpolating field can acquire endpoint phases; they need not coincide with those of the transformed field. For a Majorana field $J^\mu$ vanishes identically and $\mathcal{O}_5$ is absent. Within the displayed derivative sector, the axial operator is therefore the only dimension-five derivative channel for a Majorana spinor. Imposing $P$ or $CP$ removes the entire dimension-five derivative sector considered here. This is the Majorana-spinor mechanism relevant to the gravitino prototype, without constituting a complete classification of spin-$3/2$ effective operators.

To isolate the classical effect of the axial term, first take a constant real mass $m>0$. The axial operator $\mathcal{O}_5'$ is also a flat connection, $\lambda_A\,d\Xi$ with $\lambda_A=c_A/M_*$, now for the chiral rotations generated by $\gamma_5$, and the chiral rotation $\psi\to e^{i\lambda_A\Xi\gamma_5}\psi$ removes it from the kinetic term. The mass term, however, is not invariant: it becomes $m\,\bar\psi\,e^{2i\lambda_A\Xi\gamma_5}\psi$. The mass is therefore a section of charge two under the chiral rotations, and the gauge-invariant content of $\mathcal{O}_5'$ is not the curvature of the connection, which vanishes, but the covariant derivative of the mass section,
\begin{equation}
\big|D^{(A)}_\mu \mathcal{M}\big| = 2\,m\,|\lambda_A|\,|\partial_\mu\Xi| ,
\qquad
D^{(A)}_\mu = \partial_\mu + 2i A_\mu ,
\label{eq:mass-section}
\end{equation}
with $A_\mu=\lambda_A\partial_\mu\Xi$ and $\mathcal{M}=m$ before the rotation, and $A_\mu=0$ and $\mathcal{M}=m\,e^{2i\lambda_A\Xi}$ after it; the two frames give the same value. For a nonzero background gradient it vanishes only for $m=0$ or $\lambda_A=0$. A quantum chiral rotation must also account for anomalous contributions; $m=0$ alone does not guarantee removability in an anomalous theory. For a homogeneous background $\Xi(t)$, in the locally constant or controlled adiabatic approximation for its gradient, the axial coupling enters the Dirac Hamiltonian as $b_0\gamma_5$ with $b_0=\lambda_A\dot\Xi$, and the dispersion relation becomes
\begin{equation}
E^2 = m^2 + \left(|\mathbf{p}| + h\,b_0\right)^2 ,
\label{eq:axial-dispersion}
\end{equation}
with $h=\pm1$ the helicity. This is the structure of the axial coefficient $b_\mu$ of the Lorentz-violating Standard-Model extension~\cite{39}, here generated by the gradient of a dynamical background. The energy of a moving fermion is shifted at first order in $\lambda_A$, by $h\,b_0|\mathbf{p}|/E$, and that of a fermion at rest by $\lambda_A^2\dot\Xi^2/(2m)$, which is of the same order in $1/M_*$ as the scalar channel of Eq.~\eqref{eq:scalar}. For $m\neq0$ the axial coupling is therefore not removable.

The discrete symmetries decide between the two cases. We assume that $\Xi$ is a real Lorentz scalar even under parity and charge conjugation, and that the mass-like coupling $g_f\Xi\bar\psi\psi$ respects the corresponding discrete symmetries at the order considered. The axial current is odd under parity and even under charge conjugation, so $\mathcal{O}_5'$ is odd under $P$ and under $CP$, and it is excluded in any sector in which $P$ or $CP$ is respected at the order considered. Under this condition the axial dimension-five operator is absent; the vector operator is redundant under the current-conservation hypotheses just stated. The displayed fermionic potential $(\partial_\mu\Xi)^2\bar\psi\psi$ has dimension seven. An additional suppression of its coefficient by chirality-breaking parameters requires an approximate chiral symmetry and its breaking data; it does not follow from operator dimension alone. Without it, the leading fermionic derivative channel arises from $\mathcal{O}_5'$ and is not suppressed relative to the scalar one.

For a complex scalar, the dimension-five current operator $(\partial_\mu\Xi)\,i(\phi^\dagger\partial^\mu\phi-\partial^\mu\phi^\dagger\phi)$ can likewise be removed to first order when the corresponding current is conserved. Its removal induces terms at the same order as the dimension-six potential, however~\cite{38}. With $S_\mu=\partial_\mu\Xi$ and $J_\phi^\mu=i(\phi^*\partial^\mu\phi-\partial^\mu\phi^*\phi)$, consider
\begin{equation}
\mathcal L=|\partial\phi|^2+aS_\mu J_\phi^\mu+bS_\mu S^\mu|\phi|^2-V(|\phi|^2),
\qquad [a]=-1,\quad[b]=-2.
\label{eq:scalar-basis}
\end{equation}
The transformation $\phi=e^{ia\Xi}\chi$ gives, for these terms,
\begin{equation}
\mathcal L=|\partial\chi|^2+(b-a^2)S_\mu S^\mu|\chi|^2-V(|\chi|^2).
\label{eq:scalar-residual}
\end{equation}
The scalar derivative potential must therefore be interpreted in a reduced basis with the induced terms included. In the convention of~\eqref{eq:scalar-basis}, a complete pure-gauge connection has $b=a^2$ and leaves no residual potential. A nonzero residual coefficient generically has physical effects and cannot be discarded merely because the current operator was redundant. For a real scalar the current is absent. Equation~\eqref{eq:scalar} is understood as the displayed scalar model after this basis choice, without imposing any universal value for its independent coefficient $\kappa$. The channel comparison says nothing about the mass-noise terms or about other allowed operators outside the displayed sector.

\begin{table}[h]
\centering
\footnotesize
\begin{tabular}{lcc}
\hline\hline
\textbf{Property} & \textbf{Fermion} & \textbf{Scalar} \\
\hline
Order of equation of motion & First & Second \\
Dimension-five derivative operators & Vector redundant; axial excluded & Current redundant at first order \\
Displayed residual derivative potential & dim.\ 7 & dim.\ 6 \\
Direct derivative-channel effect & Higher operator dimension & Depends on residual coefficient \\
Ensemble dephasing from mass noise & Present & Present \\
Mediated susceptibility & Present & Present \\
\hline\hline
\end{tabular}

\caption{Comparison within the displayed operator sector and reduced basis. Fermionic entries assume the current-conservation and parity or $CP$ conditions of Section~\ref{sec:direct}; without the latter, the axial dimension-five operator can contribute. Scalar effects require a nonzero residual coefficient and the relevant fluctuations. The table does not assign the fixed-state criterion $\Lambda$ or establish a universal hierarchy. Mediated susceptibility can affect both sectors.}

\label{tab:comparison}
\end{table}

For the fermion, the linear coupling in~\eqref{eq:dirac} converts fluctuations of $\Xi$ into fluctuations of the fermion mass, $\delta m(t)=g_f\,\delta\Xi(t)$, and the ensemble-averaged fermionic amplitude then has the corresponding model-dependent scaling
\begin{equation}
\Gamma_f \propto g_f^2\,\sigma_\Xi^2\,\tau_c ,
\label{eq:gamma-f}
\end{equation}
with the proportionality factor fixed by the chosen noise correlator. For a smooth, gapped, sufficiently adiabatic realization the instantaneous Hamiltonian has sharp eigenvalues; this is not in itself a global propagator-pole theorem, and \eqref{eq:gamma-f} is inhomogeneous broadening in the sense of Section~\ref{sec:dephasing}. In geometric terms, the instantaneous Dirac Hamiltonian $H(\Xi)=\boldsymbol{\alpha}\cdot\mathbf{p}+\beta(m+g_f\Xi)$ depends on the single real parameter $\Xi$. On a contractible gapped interval the associated adiabatic connection has no curvature and can be gauged away along the path, with endpoint conventions treated consistently. In this regime the accumulated energy phase is the dynamical phase $\int E(\Xi(t))\,dt$, whose ensemble fluctuation is precisely the dephasing~\eqref{eq:gamma-f}. The structural asymmetry between~\eqref{eq:gamma-s} and~\eqref{eq:gamma-f} is therefore one of operator dimension and channel, not an absolute protection.

\begin{proposition}[Direct-channel operator hierarchy in the reduced EFT sector]
\label{prop:channel}
In the displayed effective operator sector, assume a conserved non-anomalous vector current, an admissible field redefinition, and parity or $CP$ invariance excluding the axial dimension-five coupling. Use a scalar basis in which induced dimension-six terms are retained. A nonzero residual scalar derivative potential then occurs at dimension six, whereas the displayed fermionic derivative potential occurs at dimension seven. This comparison does not establish a universal damping hierarchy: it does not exclude mass-noise dephasing, other allowed interactions, or mediated changes of the fermionic spectral measure. It is a conditional statement about the specified channel, not a general protection of fermionic particlehood.
\end{proposition}

\subsection{Mediated channel: spectral transfer}
\label{sec:mediated}

\begin{proposition}[Spectral transfer decomposition]
\label{prop:transfer}
Consider a one-loop self-energy $\Sigma_\Psi(p)$ in which a single internal line carries the dressed propagator of a mediator field $S$ with spectral measure
\begin{equation*}
d\rho_S(s) = Z_S\,\delta_{m_S^2}(ds) + d\rho_{S,\mathrm{cont}}(s) ,
\end{equation*}
where $d\rho_{S,\mathrm{cont}}$ denotes the non-atomic part relevant to the channel. At the stated order the diagram is linear in this internal propagator, with the other lines and vertices fixed. A common regulator and a subtraction prescription compatible with the spectral integration are assumed. Whenever the interchange of the spectral and loop integrations is justified, linearity gives
\begin{equation}
\Sigma_\Psi(p) = Z_S\,\Sigma_\Psi^{(0)}(p;m_S^2) + \int_{s_0}^\infty \Sigma_\Psi^{(0)}(p;s)\,d\rho_{S,\mathrm{cont}}(s) .
\label{eq:transfer}
\end{equation}
If the non-atomic part is purely absolutely continuous and has total weight $1-Z_S$, one may write $d\rho_{S,\mathrm{cont}}(s)=(1-Z_S)\hat\rho(s)\,ds$, recovering the familiar normalized continuum form. Here $\Sigma_\Psi^{(0)}(p;M^2)$ is the corresponding self-energy with a free mediator of mass $M$ in the same representation and prescription. The normalized form requires the stated sum rule; it is not universal for arbitrary renormalized or composite operators. The decomposition does not assert linearity of a fully self-consistent solution of coupled Dyson equations. Different mass-dependent subtraction conventions must be matched before using it as an identity for renormalized self-energies.
\end{proposition}

\begin{corollary}
\label{cor:transfer}
A change in the spectral measure of a mediator is transmitted to the self-energy of every field coupled through the channel represented in~\eqref{eq:transfer}. It may therefore modify the external residue, but no degradation of that residue follows from linearity alone. Whether the external atom is destroyed, reduced, shifted, or left essentially intact depends on the coupling, tensor structure, renormalization conditions, and the position of the external pole relative to the transferred continuum support.
\end{corollary}

\begin{proposition}[Threshold dichotomy with non-vanishing coupling hypotheses]
\label{prop:threshold}
Work in a physical spectral construction in which the discontinuity is supported by the intermediate channels under consideration, with no additional medium or state-induced on-shell cuts. Let $m_A^2$ denote the renormalized on-shell point, continuously connected to a pre-existing simple pole when the channel is switched on, and assume that no independent massless-infrared obstruction destroys that pole. If $m_A^2$ lies strictly below the continuum support of every admissible channel contributing to its self-energy, then
\begin{equation}
\mathrm{Im}\,\Sigma_A(m_A^2) = 0 ,
\label{eq:below}
\end{equation}
the included channels produce no on-shell absorptive part. Persistence of the isolated atom further requires a regular simple zero of the inverse propagator and a finite positive residue; these properties persist under a sufficiently small regular perturbation of a positive-residue pole as long as it remains separated from the cuts. If instead $m_A^2$ lies inside, or at the edge of, a kinematically accessible continuum support, and if the on-shell matrix element, the transferred spectral density, and the threshold phase-space factor are non-vanishing at the relevant point, then generically
\begin{equation}
\mathrm{Im}\,\Sigma_A(m_A^2) \neq 0 ,
\label{eq:inside}
\end{equation}
and the isolated atom is unstable against conversion into a resonance or another non-atomic spectral structure~\cite{40}.
\end{proposition}

\paragraph{Argument.}
The absorptive part of~\eqref{eq:transfer} at $p^2=m_A^2$ is supported where the intermediate states are kinematically accessible. Below all thresholds for the channels included in the proposition it vanishes, so those channels do not generate a width. In the scalar convention $D^{-1}(s)=s-m_0^2-\Sigma_R(s)$, pole persistence also requires a real zero $m_*^2$ and a finite positive $Z=[1-\partial_s\operatorname{Re}\Sigma_R(m_*^2)]^{-1}$. A vanishing absorptive part alone does not prove these properties. For spin or mixing, the corresponding physical projection or eigenvalue must be used. At an exact threshold, accessibility alone is not sufficient: the on-shell matrix element, the transferred spectral density, and the threshold phase-space factor must be non-vanishing. Symmetry-protected zeros, tensorial, chiral, or gauge cancellations, and threshold suppression are explicit exceptions to the generic resonance conclusion. When these factors are nonzero, the embedded case is the field-theoretic form of the instability of embedded eigenvalues in the Friedrichs model~\cite{40}.

\paragraph{Worked example of the transfer decomposition.}
The content of Propositions~\ref{prop:transfer} and~\ref{prop:threshold} is displayed by the simplest case. Let $A$ be the external field, $\chi$ a spectator with atomic spectral measure at $m_\chi^2$, and $S$ the mediator, all scalars, coupled by $\mathcal{L}\supset -g\,A\chi S$ with $[g]=1$, and let
\begin{equation*}
d\rho_S(s) = Z_S\,\delta(s-m_S^2)\,ds + (1-Z_S)\,\hat\rho(s)\,ds ,
\qquad \operatorname{supp}\hat\rho \subseteq \bigl[s_0,\infty\bigr) .
\end{equation*}
The continuum may begin below the atomic mediator mass, $s_0<m_S^2$, without support overlap alone deciding the fate of the atom: an exact embedded atom requires protection or decoupling from the overlapping continuum, whereas an unprotected unstable mediator is represented by a non-atomic peak rather than an exact delta function. Equation~\eqref{eq:transfer} is exact for the spectral measure as written; if the mediator is only a narrow resonance, the atomic term below is instead a narrow-width approximation. At one loop $\Sigma_A(p^2) = g^2\int d\rho_S(s)\,B_0(p^2;m_\chi^2,s)$, and with the standard absorptive part of the bubble, $\mathrm{Im}\,B_0(p^2;m_1^2,m_2^2) = \lambda^{1/2}(p^2,m_1^2,m_2^2)\,\theta\!\left(p^2-(m_1+m_2)^2\right)/(16\pi p^2)$, Eq.~\eqref{eq:transfer} gives
\begin{equation}
\begin{split}
\mathrm{Im}\,\Sigma_A(m_A^2) = \frac{g^2}{16\pi m_A^2}\Big[\, & Z_S\,\lambda^{1/2}(m_A^2,m_\chi^2,m_S^2)\,\theta(m_A-m_\chi-m_S) \\
& + (1-Z_S)\int_{s_0}^{(m_A-m_\chi)^2} ds\,\hat\rho(s)\,\lambda^{1/2}(m_A^2,m_\chi^2,s)\Big],
\end{split}
\label{eq:example}
\end{equation}
where $\lambda$ is the K\"all\'en function and the integral is understood to vanish when $m_A<m_\chi+\sqrt{s_0}$. For a continuum with nonzero weight in the accessible interval and a nonvanishing coupling, three regimes follow. (a) $m_A<m_\chi+\sqrt{s_0}$: both terms vanish; this is~\eqref{eq:below}, and the channel produces no width although $Z_S<1$, and a regular positive-residue pole persists under the continuation assumptions of the proposition. (b) $m_\chi+\sqrt{s_0}<m_A<m_\chi+m_S$: the atomic term vanishes but the transferred continuum does not; the degradation of $S$ has lowered the effective threshold of $A$ from $m_\chi+m_S$ to $m_\chi+\sqrt{s_0}$, and the width of $A$ is proportional to the transferred weight $1-Z_S$ and to the overlap of $\hat\rho$ with the accessible region. This is spectral transfer in its minimal form. (c) $m_A>m_\chi+m_S$: both terms contribute. The example verifies Corollary~\ref{cor:transfer} explicitly: linearity alone fixes nothing, and the absorptive contribution is decided by the external threshold $(m_\chi+\sqrt{s_0})^2$, the available spectral weight, and the coupling, while the pole and residue also require dispersive control. The threshold $s_0$ of the internal measure is not itself the external threshold. This example uses the Minkowski bubble kernel. A de Sitter channel requires its own harmonic kernel and state assumptions; Proposition~\ref{prop:dS-sieve} supplies a scalar support result in its specified domain, not a replacement of this kernel by a universal lower limit.

\section{Conditions for an asymmetric realized spectrum}
\label{sec:split}

\subsection{No split under exact supersymmetry}

\begin{proposition}[Multiplet rigidity of the criterion]
\label{prop:rigidity}
Suppose the supercharges $Q_\alpha$ exist as densely defined operators on the physical Hilbert space, annihilate the vacuum, $Q_\alpha|0\rangle=0$, and generate the $N=1$ algebra. Let $|p,B\rangle$ be a one-particle state contributing an isolated atom to the spectral measure of a field $\phi$, with $Q_\alpha|p,B\rangle\neq0$. Assume also that the supercharges act on the relevant one-particle subspace and that the chosen superpartner interpolating operators have nonzero overlap with the states they connect. Then $Q_\alpha|p,B\rangle$ is a one-particle state of the same four-momentum and mass, contributing an atom at the same location to the appropriate superpartner spectral measure. Consequently $\Lambda$ agrees on the nontrivially connected physical multiplet members.
\end{proposition}

\paragraph{Argument.}
Because $[Q_\alpha,P_\mu]=0$, acting with a conserved supercharge on a physical one-particle state preserves its four-momentum. Under the hypotheses of the proposition, the one-particle subspace therefore carries a representation of the unbroken supersymmetry algebra. If $Q_\alpha|p,B\rangle\neq0$, the resulting state belongs to the same mass multiplet. Provided the superpartner interpolating operator $[Q_\alpha,\phi]$ has non-zero overlap with that state, its two-point function contains a one-particle contribution at the same mass. Thus the existence of the isolated one-particle contribution is shared by the non-trivially related members of the multiplet. The proposition does not assert equality of field-normalization residues.

\begin{remark}
The argument assumes the existence of the supercharges on the physical state space and the invariance of the vacuum; it is subject to the usual caveats concerning massless multiplets, Goldstone sectors, and gauge artifacts. The conclusion also presupposes non-zero overlap of the chosen interpolating operators with the relevant physical one-particle states. These qualifications delimit the proposition's domain of validity.
\end{remark}

\begin{corollary}[Failure of unbroken-multiplet hypotheses is necessary for a split]
\label{cor:breaking}
A physical asymmetry between states identified as partners of an exact unbroken supermultiplet is excluded under Proposition~\ref{prop:rigidity}. A proposed split must identify which physical hypothesis fails, for example spontaneous breaking, explicit or soft breaking, or absence of the conventional global supercharges in the class of realizations specified in Proposition~\ref{prop:order}. Failure of an arbitrary interpolator to overlap with a state is not by itself physical supersymmetry breaking or evidence that the state is absent. In that setting the hypotheses of Proposition~\ref{prop:rigidity} are not available, so a split is no longer excluded by multiplet rigidity. This establishes only structural possibility: whether a split occurs, and with what magnitude, remains a dynamical and spectral question.
\end{corollary}

Corollary~\ref{cor:breaking} links the stages of the filtration: the geometry/vacuum sieve determines the realization of supersymmetry; the effective algebra inherited from it fixes which multiplet rigidities survive; kinematics and dynamics then decide the outcome field by field. At no stage is a distinction between ordinary particles and superpartners assumed as such: the physical differences must be identified in the symmetry realization, masses, couplings, or channel kinematics and then processed by the dynamical calculation.

\subsection{Ingredients of the proposed transfer-based split}
\label{sec:ingredients}

For a proposed realization based on unequal direct susceptibility and mediated transfer, the following ingredients describe the dynamical tasks. They are not a theorem listing necessary conditions for every possible split, and none follows from the filtration alone.
\begin{enumerate}
\item \emph{Seed asymmetry.} At least one class of fields must undergo direct degradation with strength differing across the two sectors, through couplings and masses rather than through sector labels.
\item \emph{Transfer network asymmetry.} The admissible interaction channels must propagate the degradation preferentially within one sector, with the reciprocal effects on the other sector computed in the same approximation.
\item \emph{Kinematic asymmetry.} The physical on-shell channels must differ in their available spectral weight and matrix elements. Being below the relevant continuum is one sufficient route to vanishing absorption (Proposition~\ref{prop:threshold}); zeros protected by symmetry or other cancellations can provide another.
\item \emph{Persistence.} For scenarios invoking an earlier high-curvature epoch, a separate dynamical analysis must determine the subsequent fate of any physical atom.
\end{enumerate}
For the scalar de Sitter channels covered by Proposition~\ref{prop:dS-sieve}, a Minkowski additive threshold is unavailable at fixed $H$. Differences in channel weights and rates remain possible, but neither an asymmetric exact spectrum nor recovery of atoms in a flat limit follows from that support statement alone.

\subsection{Persistence as a separate requirement}
\label{sec:persistence}

If a spectral atom is removed in a high-curvature or otherwise non-stationary regime, its absence at later times does not follow automatically. A constructive model must therefore determine whether the evolving spectral measure permits the re-formation of an isolated atomic component when the background approaches a lower-curvature regime. We use spectral persistence for the question of whether the non-atomic spectral status produced in one regime is maintained under subsequent background evolution, rather than an isolated atom being re-formed. This is a separate dynamical question, not a consequence of the structural filtration or of a family of invariant correlators at different fixed curvatures.

\section{Relationship with known mechanisms}
\label{sec:known}

The localization framework discussed here is distinct from Higgs mass generation, from decoupling by large mass scales, and from split-supersymmetry scenarios: those mechanisms operate on masses and couplings of asymptotic particles, whereas the present framework concerns whether an asymptotic particle exists at all. It is conceptually related to, but distinct from, infraparticle dressing and infrared decoherence~\cite{4,5,6,7,8,28,29}: the filtration organizes such effects into ordered admissibility and realization stages. Finally, through Proposition~\ref{prop:order} the framework connects to non-linear realizations of supersymmetry~\cite{19,20}: constrained multiplets can realize one effective symmetry structure on a de Sitter background, without being the unique output of the geometric restriction.

\section{Discussion and limitations}
\label{sec:discussion}

The framework has a defined domain of applicability. Geometry alone does not convert superpartners into a non-asymptotic or dark sector. Geometry and the choice of state determine the spectral environment and the realization of symmetries; interactions, thresholds, and self-energies decide the fate of individual poles. Any distinction between Standard-Model particles and superpartners must therefore be derived from masses, couplings, available channels, and supersymmetry-breaking data, rather than imposed through sector labels.

The framework also avoids treating a reduction of the residue as equivalent to the destruction of a particle. A finite $0<Z<1$ is compatible with an isolated atom. A non-zero on-shell absorptive part normally signals an unstable resonance, while failure of the atomic criterion requires the absence of a physical point-mass contribution, within the spectral and asymptotic construction being used. In genuinely time-dependent backgrounds this language is local or adiabatic and state-dependent; a global K\"all\'en--Lehmann representation need not exist.

The stochastic background model retained in Section~\ref{sec:channels} is illustrative. Ensemble dephasing of an averaged correlator is not automatically an intrinsic imaginary self-energy. A constructive realization requires an in-in or open-system calculation with the environmental degrees of freedom treated dynamically. Likewise, the spectral-transfer decomposition identifies how a dressed internal line enters a self-energy, but it does not by itself prove that the external pole disappears.

The present work establishes the structural ordering and makes explicit the domains of its spectral and symmetry statements. Its direct-channel comparison is a statement about a specified effective operator sector; the transfer decomposition identifies how a mediator measure enters a self-energy without fixing the fate of the external atom. The companion study~\cite{41} implements the dynamical side of this program in a de Sitter setting by examining residue suppression, spectral redistribution, and downstream phenomenological consequences. The structural propositions proved here do not depend on the phenomenological conclusions of that companion analysis.

These qualifications are not ancillary. They make the criterion congruent with the exterior-polynomial classification of Ref.~\cite{17}: in both papers, a restricted classification result is kept distinct from the dynamics and from the final physical realization.

\section{Conclusions}
\label{sec:conclusions}

The existence of a field in a Lagrangian, its membership in a supersymmetric algebra, and its realization as an asymptotic particle are distinct statements. The present paper organizes them through an ordered structural criterion. Geometry and the quantum state define the spectral arena; the effective algebra fixes the physically realized symmetries and channels; dynamics determines the dressed two-point function; and spectral analysis decides whether a physical atom with non-zero weight survives in the applicable construction.

This architecture is congruent with the gravitino mass-bilinear classification of Ref.~\cite{17}, which distinguishes the unrestricted vector--spinor algebra from the restricted exterior-polynomial first-order sector, separates structural classification from dynamical and supersymmetric selection, and treats superspace projection as extraction of a component already contained in a completed superform. The present work applies the same discipline to particle states.

The localization criterion $\Lambda[\Psi]$ is not an additional equation of motion. It diagnoses a physical atom of positive weight on a stated spectral domain; an isolated massive pole is its regular realization, while massless or protected point spectrum requires the appropriate treatment. The failure to construct a spectral representation is not itself proof that an atom is absent. No conclusion about an individual field follows from its background or sector label alone.

The positivity lemma is conditional on the global operator realization and the full even element of the superalgebra; superconformal extensions fall outside its assumed closure. The de Sitter scalar product result removes a Minkowski additive threshold in its domain, but does not establish the external spectral fate of every interacting field. Exponentially small curvature-induced rates do not by themselves imply recovery of positive atomic weight in a limiting measure. The worked transfer example distinguishes an internal spectral threshold from the accessible external continuum and shows why transfer alone does not decide particlehood.

The direct derivative-channel comparison is restricted to a specified effective operator sector, symmetry assumptions, and a consistently reduced basis. A removable fermionic vector connection, an excluded axial operator, and a nonzero residual scalar potential produce a genuine operator-dimension distinction in that sector, while not amounting to a universal theorem of fermionic protection. Exact unbroken multiplet structure forbids a physical split under its stated hypotheses; when these fail, the answer depends on the coupled dynamics. The companion analysis of Ref.~\cite{41} explores one explicit dynamical realization, whereas the present paper supplies the structural conditions against which such realizations should be tested.

\section*{Funding}
No specific funding was received for this work.

\section*{Conflict of interest}
The authors declare no conflict of interest.

\section*{Data availability}
No datasets were generated or analyzed in this theoretical study.

\end{document}